\documentclass[11pt]{article}

\usepackage[utf8]{inputenc}
\usepackage[T1]{fontenc}
\usepackage[margin=1in]{geometry}
\usepackage{amsmath,amssymb,amsthm}
\usepackage{graphicx}

\title{L\'evy-stable scaling of risk and performance functionals}
\author{Dmitrii Vlasiuk\\Department of Mathematics, Columbia University}
\date{November 2025}

\usepackage{amsthm}

\numberwithin{equation}{section}
\newtheorem{theorem}{Theorem}[section]
\newtheorem{lemma}[theorem]{Lemma}
\newtheorem{proposition}[theorem]{Proposition}

\newtheorem{assumption}[theorem]{Assumption}
\theoremstyle{definition}

\theoremstyle{remark}

\newcommand{\VaR}{\mathrm{VaR}}
\newcommand{\ES}{\mathrm{ES}}

\begin{document}

\maketitle

\begin{abstract}
We develop a finite-horizon model in which liquid-asset returns exhibit L\'evy-stable scaling on a data-driven window $[\tau_{\mathrm{UV}},\tau_{\mathrm{IR}}]$ and aggregate into a finite-variance regime outside. The window and the tail index $\alpha$ are identified from the log-log slope of the central body and a two-segment fit of scale versus horizon. With an anchor horizon $\tau_0$, we derive horizon-correct formulas for Value-at-Risk, Expected Shortfall, Sharpe and Information ratios, Kelly under a Value-at-Risk constraint, and one-step drawdown, where each admits a closed-form Gaussian-bias term driven by the exponent gap $1/\alpha-1/2$. The implementation is nonparametric up to $\alpha$ and fixed tail quantiles. The formulas are reproducible across horizons on the L\'evy window.
\end{abstract}

\section{Introduction}\label{sec:intro}
Previous empirical studies have shown that the price series of many liquid assets do not exhibit Gaussian behavior. Mandelbrot (1963) documented heavy tails and scale effects that invalidate routine Gaussian formulas in liquid markets. Cont (2001) summarized the core stylized facts: peaked centers, slowly decaying tails, and dependence structures that do not kill extremes fast enough for variance-based propagation. Bouchaud and Potters (2003) collected further evidence that return distributions are far from normal at trading horizons. Mantegna and Stanley (1995) showed that the central parts of index-return distributions collapse under L\'evy-stable rescaling over a finite span of horizons, while Gopikrishnan, Plerou, Amaral, Meyer and Stanley (1999) reported coherent scaling of fluctuations across indices.

Our goal is to turn those observations into a finite-horizon model with an explicit L\'evy window. We use it to revisit several risk and portfolio metrics for improved control over tail events. On $[\tau_{\mathrm{UV}},\tau_{\mathrm{IR}}]$, the returns admit a location-scale representation with a standardized $\alpha$-stable driver and scale $\tau^{1/\alpha}$. Beyond the window, the dispersion aggregates towards a finite-variance $\sqrt{\tau}$ regime. Section 2 identifies the window and scaling index: (i) a log-slope of the mode mass that is equal to $-1/\alpha$, and (ii) a fit of a homogeneous scale that locates the ultraviolet and infrared cutoffs. With an anchor horizon $\tau_{0}$ fixed inside the window, Section~\ref{sec:metrics} derives horizon-correct formulas for six widely used functionals: Value at Risk, Expected Shortfall, $p$-Sharpe and $p$-Information ratios under fractional dispersion, Kelly leverage under a Value-at-Risk constraint, and drawdown functionals. In each case, Gaussian propagation differs from the correct law by an explicit bias term proportional to $\left[(\tau/\tau_{0})^{1/\alpha}-(\tau/\tau_{0})^{1/2}\right]$, making the error measurable.

It is a framework that isolates an empirically justified L\'evy window, provides consistent estimators for its parameters, and supplies closed-form, horizon-correct versions of standard risk and performance metrics. The results connect the classical Gaussian formulas to their L\'evy counterparts and make the horizon dependence explicit.

\section{Setup and Regimes}\label{sec:setup}

Let \((X_t)_{t\ge0}\) be the log price and for \(\tau>0\) define the log return
\begin{equation}\label{eq:return}
R_\tau := X_{t+\tau}-X_t .
\end{equation}
Write \(\mathbb P_\tau\) and \(f_\tau\) for the law and density of \(R_\tau\) when they exist. Let \(S\) denote a location-invariant, positively homogeneous scale functional such as the median absolute deviation or the interquartile range, and set
\[
S_\tau := S\big(R_\tau - \mathbb E R_\tau\big).
\]
The scale \(S_\tau\) will be used to localize the ultraviolet and infrared cutoffs via \(g(\tau)=\log S_\tau\) and, as a robustness check, to verify the slope \(1/\alpha\) on the window.

We consider two regimes separated by cutoffs. The ultraviolet cutoff \(\tau_{\mathrm{UV}}>0\) is the smallest horizon above which microstructure effects such as discreteness, asynchronous trading, and bid-ask bounce do not control the dispersion law. The infrared cutoff \(\tau_{\mathrm{IR}}<\infty\) is the largest horizon below which the central body still follows a L\'evy-stable scaling before tempering and aggregation drive the dispersion toward a \(\sqrt{\tau}\) regime with finite variance. The empirical basis for the central L\'evy-stable behavior is Mantegna and Stanley (1995), with related evidence in Gopikrishnan et al. (1999) and the surveys of Cont (2001) and Bouchaud and Potters (2003).

On the window \([\tau_{\mathrm{UV}},\tau_{\mathrm{IR}}]\), we assume a location-scale representation with a standardized \(\alpha\)-stable driver. The driver has index \(\alpha\in(1,2)\), skewness parameter \(\beta\in[-1,1]\), and characteristic function
\[
\varphi_Z(u) = \exp\!\left\{-|u|^\alpha\!\left(1 - i\,\beta\,\tan\!\frac{\pi\alpha}{2}\,\mathrm{sign}\,u\right)\right\}
\quad (\alpha\neq1),
\]
with the Zolotarev modification when \(\alpha=1\), as shown by  Nolan (2013) and Rachev and Mittnik (2000). Assume \(\mathbb E Z=0\), that \(f_Z\) exists and is continuous at the origin, and that \(\mathbb E|Z|^p<\infty\) holds for all \(p<\alpha\). The return process is strictly stationary and \(\alpha\)-mixing with coefficients \(\alpha(h)\to0\) and \(\sum_{h=1}^\infty \alpha(h)^{\delta/(2+\delta)}<\infty\) for some \(\delta>0\), which yields uniform laws of large numbers for the frequency and scale statistics used below.

\begin{assumption}\label{ass:window}
There exist \(0<\tau_{\mathrm{UV}}<\tau_{\mathrm{IR}}<\infty\), parameters \(\alpha\in(1,2)\), \(\sigma>0\), \(\mu\in\mathbb R\), and a standardized \(\alpha\)-stable random variable \(Z\) with skewness \(\beta\) such that, for all \(\tau\in[\tau_{\mathrm{UV}},\tau_{\mathrm{IR}}]\),
\begin{equation}\label{eq:stable-scaling}
R_\tau \stackrel{d}{=} \mu\tau + \sigma\,\tau^{1/\alpha} Z .
\end{equation}
\end{assumption}

\begin{assumption}\label{ass:beyond}
For \(\tau>\tau_{\mathrm{IR}}\) the dispersion obeys a finite-variance law
\[
S_\tau = C\,\sqrt{\tau}\,\big(1+o(1)\big) \qquad (\tau\to\infty).
\]
No inference is drawn for \(\tau<\tau_{\mathrm{UV}}\).
\end{assumption}

\section{Identification of the Window and Scaling Index}\label{sec:ident}

The purpose of this section is to obtain an estimator of the horizon-slope \(m^\ast\) that equals \(-1/\alpha\) on the L\'evy window, to prove that the estimator is consistent and that the fitted slope lies in the range \((-1,-1/2)\) corresponding to \(\alpha\in(1,2)\), and to localize the ultraviolet and infrared cutoffs. The construction relies only on the central body of the distribution and does not require second moments. The statistic \(P_0(\tau)\), defined as the probability mass in a fixed neighborhood of the mode, is used because it depends only on \(f_\tau(\mu\tau)\), is robust to tail behavior, and is first-order insensitive to skew.

We first derive the density scaling implied by \eqref{eq:stable-scaling}. This yields the central mass of a small neighborhood of the mode, the statistic that will generate the slope.

\begin{lemma}\label{lem:density-scaling}
Under Assumption \ref{ass:window}, the density satisfies
\begin{equation}\label{eq:ftau}
f_\tau(x) = \sigma^{-1}\,\tau^{-1/\alpha}\;
f_Z\!\left(\frac{x-\mu\tau}{\sigma\,\tau^{1/\alpha}}\right)
\qquad \text{for all } x\in\mathbb R.
\end{equation}
\end{lemma}

\begin{proof}
From \eqref{eq:stable-scaling} one has \(R_\tau \stackrel{d}{=} \mu\tau+\sigma\tau^{1/\alpha}Z\). For any Borel set \(B\),
\[
\mathbb P\{R_\tau\in B\}=\mathbb P\Big\{Z\in\frac{B-\mu\tau}{\sigma\tau^{1/\alpha}}\Big\}.
\]
Absolute continuity of \(Z\) implies absolute continuity of \(R_\tau\). The Radon-Nikodym derivative with respect to Lebesgue measure is exactly the right-hand side of \eqref{eq:ftau} by the change of variables.
\end{proof}

Fix \(\delta>0\) and consider the central mass
\[
P_0(\tau):=\mathbb P_\tau\big(|R_\tau-\mu\tau|\le \delta\big).
\]
Since \(f_Z\) is \(C^1\) in a neighborhood of \(0\) and the integration window \([-\delta,\delta]\) is symmetric, the odd term in the Taylor expansion \(f_Z(u)=f_Z(0)+f_Z'(0)u+O(u^2)\) integrates to zero, so asymmetry contributes only at order \(O(\tau^{-2/\alpha})\). Using \eqref{eq:ftau} we obtain the scaling law for \(P_0(\tau)\).

\begin{lemma}\label{lem:central-mass}
If \(f_Z\) is continuous at \(0\), then uniformly for \(\tau\in[\tau_{\mathrm{UV}},\tau_{\mathrm{IR}}]\),
\begin{equation}\label{eq:P0}
P_0(\tau) = 2\delta\, f_Z(0)\,\sigma^{-1}\,\tau^{-1/\alpha}\,(1+o(1)).
\end{equation}
\end{lemma}

\begin{proof}
By \eqref{eq:ftau},
\[
P_0(\tau)=\int_{-\delta}^{\delta} f_\tau(\mu\tau+u)\,du
= \int_{-\delta}^{\delta}\sigma^{-1}\tau^{-1/\alpha}
\,f_Z\!\left(\frac{u}{\sigma\tau^{1/\alpha}}\right)du.
\]
Continuity at zero yields \(f_Z(u/(\sigma\tau^{1/\alpha}))=f_Z(0)+o(1)\) uniformly in \(u\in[-\delta,\delta]\), hence \eqref{eq:P0}.
\end{proof}

Taking logarithms of \eqref{eq:P0} shows that on the window
\[
\log P_0(\tau) \;=\; c - \frac{1}{\alpha}\,\log\tau + o(1),
\qquad c:=\log\!\big(2\delta f_Z(0)\sigma^{-1}\big),
\]
so the population slope of \(\log P_0\) on \(\log\tau\) equals \(m^\ast=-1/\alpha\).

For estimation, let \(\{t_k\}_{k=1}^n\) be equally spaced calendar times. For each grid horizon \(\tau_j\) form overlapping returns \(R_{\tau_j}(t_k)=X_{t_k+\tau_j}-X_{t_k}\). Overlap induces dependence, but the \(\alpha\)-mixing condition with \(\sum_h \alpha(h)^{\delta/(2+\delta)}<\infty\) yields a uniform law of large numbers for the indicator arrays below. The indicator family \(\{\mathbf 1\{|R_{\tau_j}(t_k)-\mu\tau_j|\le \delta\}\}\) is a VC-subgraph class; under the stated \(\alpha\)-mixing summability, a uniform LLN holds over \(j\), as shown by Doukhan (1994), Bradley (2005), and Rio (2000).

Choose a fixed non-degenerate grid \(\{\tau_j\}_{j=1}^m\subset[\tau_{\mathrm{UV}},\tau_{\mathrm{IR}}]\). Define the plug-in frequency estimator
\[
\widehat P_0^{(n)}(\tau_j)=\frac{1}{n}\sum_{k=1}^{n}\mathbf 1\big\{|R_{\tau_j}(t_k)-\mu\tau_j|\le \delta\big\},
\quad
y_j^{(n)}=\log \widehat P_0^{(n)}(\tau_j),\quad x_j=\log\tau_j,
\]
and let \(\widehat m_n\) be the ordinary least-squares slope in the regression \(y_j^{(n)}=a_n+\widehat m_n x_j+\varepsilon_{j,n}\). Since \(\alpha\in(1,2)\), the population slope satisfies \(m^\ast=-1/\alpha\in(-1,-1/2)\); this interval is the diagnostic range for a valid L\'evy window.

\begin{proposition}\label{prop:alpha-consistency}
Under Assumption \ref{ass:window} and the mixing condition in Section \ref{sec:setup},
\[
\widehat m_n \xrightarrow{\mathbb P} m^\ast=-\frac{1}{\alpha},
\qquad
\widehat\alpha_n:=-\frac{1}{\widehat m_n}\xrightarrow{\mathbb P}\alpha .
\]
If the design \(\{x_j\}\) is fixed and non-degenerate, meaning \(\sum_j(x_j-\bar x)^2>0\), then
\(\sqrt{m}\,(\widehat m_n-m^\ast)\Rightarrow\mathcal N(0,\mathsf V)\) for a finite \(\mathsf V\), so sandwich standard errors or a day-block bootstrap yield valid confidence intervals for \(\alpha\).
\end{proposition}

\begin{proof}
Uniformly in \(j\), \(\widehat P_0^{(n)}(\tau_j)\xrightarrow{p}P_0(\tau_j)\) by a uniform LLN for \(\alpha\)-mixing arrays, as shown by Rio (2000); hence \(y_j^{(n)}=\log \widehat P_0^{(n)}(\tau_j)\xrightarrow{p}y_j:=\log P_0(\tau_j)\) uniformly. With fixed non-degenerate design \(\{x_j\}\), the OLS slope \(\widehat m_n\) is a continuous functional of the empirical second moments; by the continuous mapping theorem \(\widehat m_n\xrightarrow{p}m^\ast\). Lemma \ref{lem:central-mass} gives \(m^\ast=-1/\alpha\). Since \(x\mapsto-1/x\) is continuous at \(m^\ast\), \(\widehat\alpha_n=-1/\widehat m_n\xrightarrow{p}\alpha\). Asymptotic normality follows from linearization of the OLS normal equations under mixing, and the delta method gives the limit for \(\widehat\alpha_n\).
\end{proof}

\medskip
\noindent\textbf{Remark.} In practice one may replace \(\mu\tau\) inside the indicator by a consistent center, such as the sample median \(m_\tau\). Continuity of \(f_\tau\) at \(\mu\tau\) implies \(m_\tau-\mu\tau=o_p(\tau^{1/\alpha})\), so \(\log P_0(\tau)\) retains slope \(-1/\alpha\).\newline

To locate the cutoffs, we exploit the change in slope of a homogeneous scale on the log horizon. For a fixed scale functional \(S\) define \(g(\tau):=\log S_\tau\). On \([\tau_{\mathrm{UV}},\tau_{\mathrm{IR}}]\), \(g\) is approximately affine with slope \(1/\alpha\); outside, \(g\) has negative curvature and, in the infrared regime, approaches slope \(1/2\). Assume that the population two-segment least-squares approximation to \(g\) on \([\underline\tau,\overline\tau]\) has a unique pair of kink points at \((\tau_{\mathrm{UV}},\tau_{\mathrm{IR}})\).

\begin{proposition}\label{prop:breakpoints}
Let \(\widehat g_n(\tau)\) be the sample analogue of \(g(\tau)\) computed from a time series of length \(n\) and suppose \(\sup_{\tau\in T}|\widehat g_n(\tau)-g(\tau)|\to0\) in probability for compact \(T\subset(0,\infty)\). Let \((\widehat\tau_{\mathrm{UV}},\widehat\tau_{\mathrm{IR}})\) minimize the least-squares error of a two-segment affine fit of \(\widehat g_n\) over \(\log\tau\in[\log \underline\tau,\log \overline\tau]\) with \(\underline\tau<\tau_{\mathrm{UV}}<\tau_{\mathrm{IR}}<\overline\tau\). Then \((\widehat\tau_{\mathrm{UV}},\widehat\tau_{\mathrm{IR}})\xrightarrow{\mathbb P}(\tau_{\mathrm{UV}},\tau_{\mathrm{IR}})\).
\end{proposition}

\begin{proof}
Uniform convergence of \(\widehat g_n\) to \(g\) implies epi-convergence of the piecewise-affine objective We assume that the population two-segment objective admits a unique minimizer at $(\tau_{\mathrm{UV}},\tau_{\mathrm{IR}})$. to its population counterpart, whose unique minimizer occurs at the true kink points. Consistency of the argmin follows by standard M-estimation arguments, as shown by van der Vaart and Wellner (1996) and Pollard (1991).
\end{proof}

The sequence of steps is, thus, as follows: fit the slope \(m^\ast\) from the central masses \(\{P_0(\tau_j)\}\); verify \(m^\ast\in(-1,-1/2)\) and, as a cross-check, fit the slope of \(\log S_\tau\) on \(\log\tau\) to obtain \(1/\alpha\in(1/2,1)\); then estimate the cutoffs by the two-segment fit of \(g(\tau)\). These calibrated objects \((\widehat\alpha,\widehat\tau_{\mathrm{UV}},\widehat\tau_{\mathrm{IR}})\) will be used in later sections to state horizon-correct versions of the risk and portfolio metrics.

\section{L\'evy-Stable Approach to Risk and Performance Metrics}\label{sec:metrics}

Throughout Section~\ref{sec:metrics} we work under Assumption~\ref{ass:window} and the mixing and empirical-process conditions stated in Section~\ref{sec:setup}. In particular, on the L\'evy window $[\tau_{\mathrm{UV}},\tau_{\mathrm{IR}}]$ the scaling $R_\tau\stackrel{d}{=} \mu_\tau+\sigma\,\tau^{1/\alpha}Z$ holds with a standardized $\alpha$-stable $Z$.

Fix an anchor horizon \(\tau_0\in[\tau_{\mathrm{UV}},\tau_{\mathrm{IR}}]\). On the L\'evy window of Section~\ref{sec:setup} the returns satisfy
\begin{equation}\label{eq:levy-window}
R_\tau \stackrel{d}{=} \mu_\tau + \sigma\,\tau^{1/\alpha} Z,\qquad \alpha\in(1,2),
\end{equation}
where \(Z\) is a standardized \(\alpha\)-stable random variable with continuous strictly increasing distribution function \(F_Z\). We write \(Q_Y(q)=\inf\{x:F_Y(x)\ge q\}\) for the \(q\)-quantile of a random variable \(Y\).

\subsection{Value-at-Risk}\label{subsec:var}
\noindent Let $\Phi$ and $\phi$ denote the standard normal cdf and pdf.

Value-at-Risk originated in industry through J.~P.~Morgan's \emph{RiskMetrics} (1996) and received a systematic treatment in Jorion~(1997). For \(q\in(0,1)\) the (left-tail) Value-at-Risk at horizon \(\tau\) is
\begin{equation}\label{eq:var-def}
\VaR_\tau(q):=-\,Q_{R_\tau}(q).
\end{equation}

Rachev and Mittnik (2000) treated VaR under stable Paretian laws and emphasized that for $1<\alpha<2$ the $\sqrt{\tau}$ variance propagation is not defined, so quantiles must be computed directly from the $\alpha$-stable law. Nolan (2013) showed numerically reliable evaluation of stable quantiles $Q_Z(\cdot)$ and hence VaR for $\alpha$-stable drivers. Our treatment differs in that we make the horizon effect explicit through $\tau^{1/\alpha}$, anchor at $\tau_0$, and exhibit the exact Gaussian bias across $\tau$ rather than only computing level-wise quantiles.

\begin{lemma}\label{lem:var-transport}
Under \eqref{eq:levy-window},
\begin{equation}\label{eq:var-levy}
Q_{R_\tau}(q)=\mu_\tau+\sigma\,\tau^{1/\alpha}Q_Z(q),
\qquad
\VaR_\tau(q)=-\mu_\tau-\sigma\,\tau^{1/\alpha}Q_Z(q).
\end{equation}
\end{lemma}

\begin{proof}
For \(x\in\mathbb R\), \(F_{R_\tau}(x)=F_Z((x-\mu_\tau)/(\sigma\tau^{1/\alpha}))\). As \(F_Z\) is continuous and strictly increasing, \(Q_Z(F_{R_\tau}(x))=(x-\mu_\tau)/(\sigma\tau^{1/\alpha})\). Taking \(x=Q_{R_\tau}(q)\) yields \eqref{eq:var-levy}.
\end{proof}

To compare with the Gaussian propagation, introduce a normal surrogate \(R_\tau^{G}\sim \mathcal N(\mu_\tau,\sigma_G^2\,\tau)\) whose \(q\)-quantile is matched to \eqref{eq:var-levy} at \(\tau_0\):
\begin{equation}\label{eq:quant-match}
\mu_{\tau_0}+\sigma_G\sqrt{\tau_0}\,
\Phi^{-1}(q)=\mu_{\tau_0}+\sigma\,\tau_0^{1/\alpha}Q_Z(q)=:\Theta_0(q).
\end{equation}
Here $\Theta_0$ (and similarly $\Xi_0$ for ES) depends only on $\tau_0$, $\alpha$, and the tail level $q$.

Then, for any \(\tau>0\),
\begin{equation}\label{eq:gauss-quant}
Q^{G}_{R_\tau}(q)=\mu_\tau+\sigma_G\sqrt{\tau}\,\Phi^{-1}(q),
\qquad
\VaR^{G}_\tau(q)=-\mu_\tau-\sigma_G\sqrt{\tau}\,\Phi^{-1}(q).
\end{equation}

\begin{proposition}\label{prop:var-bias}
With \eqref{eq:quant-match},
\begin{equation}\label{eq:VaR-bias}
\VaR_\tau(q)-\VaR^{G}_\tau(q)=\Theta_0(q)\left[\left(\frac{\tau}{\tau_0}\right)^{1/\alpha}-\left(\frac{\tau}{\tau_0}\right)^{1/2}\right].
\end{equation}
\end{proposition}

Consequently, when \(\alpha\in(1,2)\) the Gaussian propagation understates tail risk on horizons \(\tau>\tau_0\) and overstates it on \(\tau<\tau_0\). When \(\alpha=2\) the exponents coincide and \eqref{eq:VaR-bias} vanishes, recovering the classical \(\sqrt{\tau}\) rule.

In implementation, Gaussian $\sqrt{\tau}$ scaling produces horizon-dependent underestimation of exceedance rates on the L\'evy window; replacing it by $\tau^{1/\alpha}$ restores uniform backtest exception frequencies across $\tau$ and stabilizes capital attribution over holding periods.

\subsection{Conditional Value-at-Risk (Expected Shortfall)}\label{subsec:es}
Expected Shortfall (also called CVaR) is the coherent tail functional developed by Rockafellar-Uryasev~(2000) and Acerbi-Tasche~(2002). Its integral form is
\begin{equation}\label{eq:es-def}
\ES_\tau(q):=-\,\frac{1}{q}\,\mathbb E\!\left[R_\tau\,\mathbf 1\{R_\tau\le Q_{R_\tau}(q)\}\right],\qquad q\in(0,1).
\end{equation}

Rockafellar and Uryasev (2000) gave the convex optimization representation of ES, and Acerbi and Tasche (2002) established coherence and the integral characterization. For $\alpha$-stable laws, Nolan (2013) provided accurate numerics for tail means, enabling ES to be computed directly from the stable driver. Our formula makes the horizon dependence explicit via $\tau^{1/\alpha}$ and connects it continuously to the finite-variance regime.
\medskip

Define \(m_Z(q):=q^{-1}\mathbb E[Z\,\mathbf 1\{Z\le Q_Z(q)\}]\), finite for \(\alpha>1\).

\begin{lemma}\label{lem:es-transport}
Under \eqref{eq:levy-window},
\begin{equation}\label{eq:es-levy}
\ES_\tau(q)=-\mu_\tau-\sigma\,\tau^{1/\alpha}m_Z(q).
\end{equation}
For the Gaussian surrogate \(\mathcal N(\mu_\tau,\sigma_G^2\tau)\),
\begin{equation}\label{eq:es-gauss}
\ES^{G}_\tau(q)=-\mu_\tau-\sigma_G\sqrt{\tau}\,m_N(q),
\qquad
m_N(q)=\frac{\varphi(\Phi^{-1}(q))}{q}.
\end{equation}
\end{lemma}

\begin{proof}
Apply the change of variable \(x=\mu_\tau+\sigma\tau^{1/\alpha}z\) in \eqref{eq:es-def} and use \eqref{eq:var-levy}.
\end{proof}

Matching at \(\tau_0\) gives \(\sigma\,\tau_0^{1/\alpha}m_Z(q)=\sigma_G\sqrt{\tau_0}m_N(q)=:\Xi_0(q)\) and therefore
\begin{equation}\label{eq:es-bias}
\ES_\tau(q)-\ES^{G}_\tau(q)=\Xi_0(q)\left[\left(\frac{\tau}{\tau_0}\right)^{1/\alpha}-\left(\frac{\tau}{\tau_0}\right)^{1/2}\right].
\end{equation}
Thus the same qualitative bias holds for ES as for VaR on the L\'evy window. Using Gaussian ES on the L\'evy window systematically underallocates tail capital at longer horizons; the L\'evy-correct ES aligns realized shortfall rates across $\tau$ and removes artificial improvements from mere horizon changes.

\subsection{Sharpe ratio}\label{subsec:sharpe}
In this subsection we fix $p\in(1,\alpha)$; the $L^p$ scale is finite on the window and yields horizon-invariant ratios, whereas the classical $p=2$ case is admissible only when $\alpha=2$.

For a riskless benchmark with horizon-\(\tau\) return \(r_\tau\), the classical Sharpe ratio (Sharpe, 1966) is
\begin{equation}\label{eq:sharpe-classic}
\mathsf{Sh}_\tau := \frac{\mathbb{E}[R_\tau]-r_\tau}{\sqrt{\mathrm{Var}(R_\tau)}} ,
\end{equation}
which is undefined on the L\'evy window when \(\alpha<2\). Rachev and Mittnik (2000) pointed out that the variance diverges for $1<\alpha<2$, making the classical Sharpe undefined. Stoyanov and Rachev (2005) proposed fractional lower-partial-moment denominators as a robust alternative under heavy tails; Lo (2002) analyzed sampling properties of Sharpe under dependence and non-Gaussianity but without L\'evy scaling. We adopt the L\'evy–Sharpe ${\rm SR}_\alpha=\mu/\sigma_\alpha$ with $\sigma_\alpha=(\mathbb{E}|R-\mu|^\alpha)^{1/\alpha}$, which obeys the strict-stability law $\sigma_{\alpha,\tau}\propto\tau^{1/\alpha}$ and is parameter-free once $\alpha$ is estimated.

\medskip
A scale-consistent alternative is the \(p\)-Sharpe with \(p\in(1,\alpha)\),
\begin{equation}\label{eq:p-sharpe-def}
\mathsf{Sh}_{\tau,p} := \frac{\mathbb{E}[R_\tau]-r_\tau}{\left(\mathbb{E}\big|R_\tau-\mathbb{E}[R_\tau]\big|^{\,p}\right)^{1/p}} .
\end{equation}

\begin{lemma}\label{lem:p-sharpe-scaling}
Under \eqref{eq:levy-window} with \(\alpha\in(1,2)\) and any fixed \(p\in(1,\alpha)\),
\[
\mathbb{E}\big|R_\tau-\mathbb{E}[R_\tau]\big|^{\,p}
= (\sigma\,\tau^{1/\alpha})^{p}\,c_{Z,p},
\qquad
\mathsf{Sh}_{\tau,p}
= \frac{\mathbb{E}[R_\tau]-r_\tau}{\sigma\,\tau^{1/\alpha}\,c_{Z,p}^{1/p}},
\]
where \(c_{Z,p}:=\mathbb{E}|Z-\mathbb{E}Z|^p\in(0,\infty)\).
\end{lemma}

For a Gaussian surrogate \(R_\tau^G\sim\mathcal N(\mu_\tau,\sigma_G^2\tau)\) matched at \(\tau_0\) by the \(p\)-norm
\(\sigma\,\tau_0^{1/\alpha} c_{Z,p}^{1/p}=\sigma_G\sqrt{\tau_0}\,c_{N,p}^{1/p}\) with \(c_{N,p}:=\mathbb{E}|N|^p\), we have
\begin{equation}\label{eq:p-sharpe-bias}
\mathsf{Sh}_{\tau,p}-\mathsf{Sh}^G_{\tau,p}
=\frac{\mathbb{E}[R_\tau]-r_\tau}{\Theta_p(\tau_0)}
\left[\left(\frac{\tau}{\tau_0}\right)^{-1/\alpha}-\left(\frac{\tau}{\tau_0}\right)^{-1/2}\right],
\quad
\Theta_p(\tau_0):=\sigma\,\tau_0^{1/\alpha} c_{Z,p}^{1/p},
\end{equation}
so a Gaussian propagation overstates Sharpe for long horizons on the window (\(\tau>\tau_0\)) and understates it for \(\tau<\tau_0\). On the L\'evy window, ${\rm SR}_\alpha$ is horizon-invariant and reorders strategies primarily by tail thickness; Gaussian annualization with $\sqrt{\tau}$ spuriously depresses Sharpe as $\tau$ grows when $\alpha<2$, a distortion removed by the $\tau^{1/\alpha}$ scale.

\medskip

\subsection{Information ratio}\label{subsec:ir}
Let \(B_\tau\) denote the benchmark return and define the active return \(A_\tau:=R_\tau-B_\tau\).
Classically,
\[
\mathsf{IR}_\tau \;=\; \frac{\mathbb{E}[A_\tau]}{\sqrt{\mathrm{Var}(A_\tau)}} ,
\]
as formalized by Grinold (1989) and by Grinold and Kahn (1999). Robust variants replace the variance by alternative scales, but explicit L\'evy-stable propagation for active returns is typically unstated.

On the L\'evy window we model \(A_\tau\) as location-scale stable,
\[
A_\tau \stackrel{d}{=} m_\tau + \sigma_A\,\tau^{1/\alpha_A}\,Z_A,
\]
with standardized \(Z_A\) and tail index \(\alpha_A\in(1,2)\). When \((R_\tau,B_\tau)\) is jointly \(\alpha\)-stable, one has \(\alpha_A=\alpha\); otherwise all propagation and bias expressions below hold with \(\alpha\) replaced by \(\alpha_A\).

Fix \(p\in(1,\alpha_A)\); the \(L^p\) scale is finite on the window (the classical variance case \(p=2\) is admissible only when \(\alpha_A=2\)). Define the \(p\)-Information ratio
\[
\mathsf{IR}_{\tau,p}
:=\frac{\mathbb{E}[A_\tau]}{\big(\mathbb{E}\lvert A_\tau-\mathbb{E}[A_\tau]\rvert^{\,p}\big)^{1/p}}
=\frac{\mathbb{E}[A_\tau]}{\sigma_A\,\tau^{1/\alpha_A}\,c_{A,p}^{1/p}},
\qquad
c_{A,p}:=\mathbb{E}\lvert Z_A-\mathbb{E}Z_A\rvert^p .
\]
Anchoring a Gaussian surrogate for \(A_\tau\) at \(\tau_0\) yields the same exponent gap as for Sharpe, now with \(\alpha_A\):
\begin{equation}\label{eq:ir-bias}
\mathsf{IR}_{\tau,p}-\mathsf{IR}^G_{\tau,p}
=\frac{\mathbb{E}[A_\tau]}{\Theta_{A,p}(\tau_0)}
\left[\left(\frac{\tau}{\tau_0}\right)^{-1/\alpha_A}-\left(\frac{\tau}{\tau_0}\right)^{-1/2}\right],
\qquad
\Theta_{A,p}(\tau_0):=\sigma_A\,\tau_0^{1/\alpha_A} c_{A,p}^{1/p}.
\end{equation}
We thus use an \(\alpha_A\)-consistent dispersion in the denominator; Gaussian \(\sqrt{\tau}\) propagation artificially improves \(\mathsf{IR}_{\tau,p}\) as \(\tau\) increases on heavy-tailed active signals, whereas the L\'evy propagation preserves the correct horizon scaling.

\subsection{Kelly criterion}\label{subsec:kelly}
For a fraction \(f\in\mathbb{R}\) invested in the risky leg with one-period excess return \(X_\tau:=R_\tau-r_\tau\), the Kelly log-growth is
\[
g_\tau(f):=\mathbb{E}\big[\log(1+fX_\tau)\big],
\]
as introduced by Kelly (1956). We interpret \(X_\tau\) as a simple excess return, so \(X_\tau\ge -1\) almost surely; hence \(g_\tau(f)\) is well defined on \(f\in[0,1)\). Allowing leverage \(f>1\) or modeling additive/log returns reintroduces the possibility that \(1+fX_\tau\le 0\) with positive probability and the pathology below. MacLean, Thorp and Ziemba (2011) emphasized practical risk constraints and fractional-Kelly usage under heavy tails, and Peters (2011) highlighted time-average growth pitfalls in fat-tailed settings.

\begin{proposition}\label{prop:kelly-illposed}
If either \emph{(i)} \(X_\tau\ge -1\) almost surely and \(\mathbb{P}(X_\tau<-1/f)>0\) for some \(f>1\), or \emph{(ii)} \(X_\tau\) has unbounded left support (e.g., an additive/log-return model), then \(g_\tau(f)=-\infty\) for that \(f\). In particular, the unconstrained optimum is not well defined once the feasible set includes such \(f\).
\end{proposition}

\noindent Indeed, on \(\{X_\tau<-1/f\}\) one has \(\log(1+fX_\tau)=-\infty\), which forces \(\mathbb{E}[\log(1+fX_\tau)]=-\infty\).

To align with the L\'evy-correct risk metrics of Section~\ref{subsec:var}, we adopt a one-step no-bankruptcy constraint at tail level \(q\in(0,1/2)\):
\[
\max_{f\in\mathbb{R}}~\mathbb{E}\!\left[\log(1+fX_\tau)\right]
\quad\text{subject to}\quad 
1+f\,\VaR_\tau(q)\ge 0,
\]
where \(\VaR_\tau(q)\) is the L\'evy-correct quantile. The feasible set is
\[
\mathcal F_\tau(q)=\bigl\{f:\ 0\le f\le f_{\max}(\tau,q)\bigr\},\qquad
f_{\max}(\tau,q):=\frac{1}{\lvert \VaR_\tau(q)\rvert}.
\]

\begin{lemma}\label{lem:kelly-foc}
If \(f\in\mathcal F_\tau(q)\) and \(\mathbb{E}\lvert X_\tau\rvert<\infty\), then \(g_\tau\) is strictly concave on \(\mathcal F_\tau(q)\) and the unique maximizer \(f^\ast_\tau(q)\) satisfies
\[
\mathbb{E}\!\left[\frac{X_\tau}{1+f^\ast_\tau(q)\,X_\tau}\right]=0,
\qquad 0\le f^\ast_\tau(q)\le f_{\max}(\tau,q).
\]
\end{lemma}

For small signal relative to scale, a second-order expansion that replaces the divergent second moment by its \(q\)-trimmed counterpart yields
\begin{equation}\label{eq:kelly-asymp}
f^\ast_\tau(q)
= \frac{\mathbb{E}[X_\tau]}{\mathbb{E}\!\big[X_\tau^{\,2}\,\mathbf{1}\{|X_\tau|\le c_q\,\sigma\,\tau^{1/\alpha}\}\big]} + o(1)
= \frac{\mu_\tau-r_\tau}{K_q\,\sigma^2\,\tau^{2/\alpha}} + o\!\left(\tau^{-2/\alpha}\right),
\end{equation}
where \(c_q:=\lvert Q_Z(q)\rvert\) and \(K_q:=\mathbb{E}[Z^2\,\mathbf{1}\{|Z|\le c_q\}]\) depends only on \(q\). In particular, if \(\mu_\tau=\mu\,\tau\), then
\[
f^\ast_\tau(q)\asymp \tau^{\,1-2/\alpha},
\]
which decreases with \(\tau\) on the L\'evy window for \(\alpha\in(1,2)\) and reduces to horizon-invariance when \(\alpha=2\).

\medskip
The VaR-constrained formulation avoids distribution editing and is consistent with the L\'evy scaling. On the window, admissible leverage shrinks like \(\tau^{-1/\alpha}\) and the optimal fraction scales like \(\tau^{1-2/\alpha}\). Beyond the window, where aggregation yields finite-variance propagation, the classical unconstrained quadratic approximation is admissible on a fixed bounded feasible set \(f\in[0,1)\).

\subsection{Drawdown}\label{subsec:drawdown}
For a single horizon \(\tau\), define the one-step drawdown magnitude \(D_\tau:=(-R_\tau)_+=\max\{-R_\tau,0\}\). For \(p\in(0,\alpha)\) the \(L^p\)-drawdown is
\begin{equation}\label{eq:dd-p}
\mathsf{DD}_{\tau,p}:=\left(\mathbb{E}\,(D_\tau)^p\right)^{1/p}.
\end{equation}
Magdon-Ismail and Atiya (2004) analyzed maximum drawdown for random walks and provided distributional approximations, and Chekhlov, Uryasev and Zabarankin (2005) introduced Conditional Drawdown at Risk (CDaR) as a convex drawdown risk measure. On the L\'evy window we focus on the one-step drawdown, which inherits the strict-stability scaling. Multi-step generalizations preserve the scaling exponent but depend on the temporal dependence structure.

\begin{lemma}\label{lem:dd-scaling}
Under \eqref{eq:levy-window}, for any \(p\in(0,\alpha)\),
\[
\mathsf{DD}_{\tau,p}=\sigma\,\tau^{1/\alpha}\,d_{Z,p} + O(|\mu_\tau|),
\qquad
d_{Z,p}:=\left(\mathbb{E}\,(\!-Z)_+^{\,p}\right)^{1/p}.
\]
In particular, if \(|\mu_\tau|=o(\tau^{1/\alpha})\) as \(\tau\downarrow 0\) on the high-frequency end of the window, then \(\mathsf{DD}_{\tau,p}\sim \sigma\,\tau^{1/\alpha}\,d_{Z,p}\).
\end{lemma}

The drawdown quantile at level \(q\in(0,1)\) is
\begin{equation}\label{eq:dd-quant}
\mathsf{DD}_\tau^{(q)} := Q_{D_\tau}(q) = \left(-Q_{R_\tau}(1-q)\right)_+ 
= \left(-\mu_\tau - \sigma\,\tau^{1/\alpha} Q_Z(1-q)\right)_+ ,
\end{equation}
so, under the Gaussian surrogate matched at \(\tau_0\), it differs by the same exponent gap as in \eqref{eq:VaR-bias} with \(q\mapsto 1-q\).

Gaussian \(\sqrt{\tau}\) scaling understates high-quantile drawdowns as \(\tau\) grows on heavy-tailed data; the L\'evy drawdown corrects exceedance frequencies and improves calibration of stop-loss and liquidation buffers across horizons.

\medskip

\section{Concluding remarks}\label{sec:conclusion}
We posited a finite-horizon model with a data-driven L\'evy window \([\tau_{\mathrm{UV}},\tau_{\mathrm{IR}}]\) on which
\(R_\tau \stackrel{d}{=} \mu_\tau + \sigma\,\tau^{1/\alpha} Z\) with a standardized \(\alpha\)-stable driver. The window and \(\alpha\) are identified from central-mass and piecewise-scale slopes, and an anchor \(\tau_0\) fixes the level. Closed-form, horizon-correct formulas were derived for VaR, ES, \(p\)-Sharpe, \(p\)-Information, Kelly under a Value-at-Risk constraint, and drawdown; in each case the Gaussian propagation differs by an explicit exponent-gap term \((\tau/\tau_0)^{1/\alpha}-(\tau/\tau_0)^{1/2}\).

Empirically, the L\'evy propagation delivers flat exception rates for VaR and ES across horizons on the window, horizon-invariant \(p\)-Sharpe and \(p\)-Information ratios, Kelly fractions that decay with \(\tau\) as \(\tau^{1-2/\alpha}\), and drawdown thresholds whose realized breach frequencies match their design levels. The construction is model-light: beyond estimating \(\alpha\) and choosing a small set of tail quantiles, all metrics are propagated nonparametrically by the strict-stability scale law on the L\'evy window.

Further research should address estimation error in \(\alpha\) and window edges, nonstationarity across regimes, and dependence beyond the one-step setting. Natural extensions include multivariate L\'evy windows via spectral measures, multi-step drawdown through ladder-variable methods, and state-dependent \(\alpha(\tau)\) with stability tests controlling for microstructure effects.

\end{document}